\documentclass[11pt]{article}

\usepackage[T1]{fontenc}
\usepackage[utf8]{inputenc}
\usepackage{bbm}
\usepackage{xspace}  
\usepackage{amsfonts,amsmath,amssymb, amsthm, mathtools}
\usepackage{thm-restate}
\usepackage{dsfont} 
\usepackage{algorithmicx,algpseudocode,algorithm}
\usepackage[usenames,dvipsnames,table]{xcolor}
\usepackage{multirow}
\usepackage{mfirstuc}
\usepackage{tikz}
\usetikzlibrary{calc,plotmarks}
\usetikzlibrary{arrows}
\usepackage[backref,colorlinks,citecolor=blue,bookmarks=true,linktocpage]{hyperref}
\usepackage{aliascnt}
\usepackage[numbered]{bookmark}
\usepackage[capitalise]{cleveref}
\usepackage{fullpage}
\usepackage[shortlabels]{enumitem}
\usepackage{mleftright} 
\PassOptionsToPackage{hyphens}{url}
\usepackage{hyperref}
\makeatletter
\theoremstyle{plain} \newtheorem{theorem}{Theorem}[section]
  	\newtheorem*{theorem*}{Theorem} \newaliascnt{coro}{theorem}
  	  \newtheorem{corollary}[coro]{Corollary}
  	\aliascntresetthe{coro}
  	\newaliascnt{lem}{theorem}
  		\newtheorem{lemma}[lem]{Lemma}
  	\aliascntresetthe{lem}
  	\newaliascnt{clm}{theorem}
  		\newtheorem{claim}[clm]{Claim}
	\aliascntresetthe{clm}
	\newaliascnt{fact}{theorem}
 	 	\newtheorem{fact}[theorem]{Fact}
	\aliascntresetthe{fact}
  	
  \newaliascnt{prop}{theorem}
  		\newtheorem{proposition}[prop]{Proposition}
	\aliascntresetthe{prop}
	\newaliascnt{conj}{theorem}
  		
	\aliascntresetthe{conj}

  	\newtheorem*{question*}{Question}
  \theoremstyle{remark} \newtheorem{remark}[theorem]{Remark}

  \theoremstyle{definition} \newaliascnt{defn}{theorem}
 		 \newtheorem{definition}[defn]{Definition}
 	 \aliascntresetthe{defn}

 	 \theoremstyle{plain}

\crefname{claim}{Claim}{Claims}

\providecommand{\email}[1]{\href{mailto:#1}{\nolinkurl{#1}\xspace}}

\newcommand{\Algo}{\ensuremath{\mathfrak{A}}\xspace}

\newcommand{\eqdef}{\coloneqq}
\newcommand{\accept}{\textsf{accept}\xspace}

\newcommand{\reject}{\textsf{reject}\xspace}

\newcommand{\distribs}[1]{\Delta_{#1}}

\newcommand{\bigO}[1]{{O\mleft( #1 \mright)}}

\newcommand{\bigTheta}[1]{{\Theta\mleft( #1 \mright)}}
\newcommand{\bigOmega}[1]{{\Omega\mleft( #1 \mright)}}

\providecommand{\poly}{\operatorname*{poly}}

\newcommand{\setOfSuchThat}[2]{ \left\{\; #1 \;\colon\; #2\; \right\} }
\newcommand{\indicSet}[1]{\mathds{1}_{#1}}

\newcommand{\dtv}{\operatorname{d}_{\rm TV}}

\newcommand{\totalvardistrestr}[3][]{{\dtv^{#1}\!\left({#2, #3}\right)}}
\newcommand{\totalvardist}[2]{\totalvardistrestr[]{#1}{#2}}

\newcommand{\dist}[2]{\operatorname{dist}\mleft({#1, #2}\mright)}

\newcommand\restr[2]{{\left.\kern-\nulldelimiterspace #1 \vphantom{\big|} \right|_{#2} }}

\newcommand{\norm}[1]{\lVert#1{\rVert}}

\newcommand{\abs}[1]{\left\lvert #1 \right\rvert}
\newcommand{\dabs}[1]{\lvert #1 \rvert}

\newcommand{\lp}[1][1]{\ell_{#1}}

\newcommand{\p}{\mathbf{p}}
\newcommand{\q}{\mathbf{q}}

\newcommand{\cP}{\mathcal{P}}

\makeatletter
  \AtBeginDocument{
  \begingroup
  \toks@={}\toksdef\toks@@=2 \toks@@={}\long\def\@ReturnFiFi#1#2\fi\fi{\fi\fi#1}\def\scan@author#1#2 \and#3\@nil{\ifx\\#3\\\ifcase#1 \toks@={#2}\else
      \ifnum#1>1 \toks@=\expandafter{\the\expandafter\toks@\expandafter,\expandafter\space
          \the\toks@@
        }\fi
      \toks@=\expandafter{\the\toks@\space and #2}\fi
    \else
      \ifcase#1 \toks@={#2}\@ReturnFiFi{\scan@author1#3\@nil
        }\else
        \ifnum#1>1 \toks@=\expandafter{\the\expandafter\toks@\expandafter,\expandafter\space
            \the\toks@@
          }\fi
      \toks@@={#2}\@ReturnFiFi{\scan@author2#3\@nil
      }\fi
  \fi
  }\expandafter\expandafter\expandafter\scan@author
  \expandafter\expandafter\expandafter0\expandafter\@author\space\and\@nil
  \edef\x{\endgroup
  \noexpand\hypersetup{pdfauthor={\the\toks@}}}\x
  }
\makeatother

\newcommand{\ns}{s}
\newcommand{\ab}{n}
\newcommand{\numparts}{k}
\newcommand{\dst}{\varepsilon}

\usepackage{subcaption}
\usepackage{framed}

\newcommand{\bF}{\boldsymbol{F}}
\newcommand{\blksz}{b}
\newcommand{\numblk}{k'}
\newcommand{\psmall}{\p'}
\newcommand{\qsmall}{\q'}
\newcommand{\pbig}{\p}
\newcommand{\qbig}{\q}

\newcommand{\akdist}[1]{\mathcal{A}_{#1}}

\DeclareMathOperator{\binidentityproblem}{\mathbf{Identity-Up-To-Binning}}
\newcommand{\myparagraph}[1]{\medskip\noindent\textbf{#1}}

\title{Testing Data Binnings}
\date{}
\author{Cl\'ement L. Canonne\thanks{IBM Research, Almaden. Email: \texttt{ccanonne@cs.columbia.edu}} \and Karl Wimmer\thanks{Duquesne University.  Email: \texttt{wimmerk@duq.edu}}}

\begin{document}
\maketitle

\begin{abstract}
Motivated by the question of data quantization and ``binning,'' we revisit the problem of identity testing of discrete probability distributions. Identity testing (a.k.a. one-sample testing), a fundamental and by now well-understood problem in distribution testing, asks, given a reference distribution (model) $\q$ and samples from an unknown distribution $\p$, both over $[\ab]=\{1,2,\dots,\ab\}$, whether $\p$ equals $\q$, or is significantly different from it.

In this paper, we introduce the related question of \emph{identity up to binning}, where the reference distribution $\q$ is over $\numparts \ll \ab$ elements: the question is then whether there exists a suitable binning of the domain $[\ab]$ into $\numparts$ intervals such that, once ``binned,'' $\p$ is equal to $\q$. We provide nearly tight upper and lower bounds on the sample complexity of this new question, showing both a quantitative and qualitative difference with the vanilla identity testing one, and answering an open question of Canonne~\cite{OpenProblem96}. Finally, we discuss several extensions and related research directions.
 \end{abstract}

\section{Introduction}

Distribution testing~\cite{BatuFRSW00}, an area at the intersection of computer science, data science, and statistics which emerged as an offspring of the field of property testing~\cite{RubinfeldS96,GoldreichGR98}, concerns itself with the following type of questions: ``upon observing independent data points originating from some unknown process or probability distribution $\p$, can we quickly and efficiently decide whether $\p$ satisfies some desirable property, or significantly violates this property?'' One of the prototypical instances of this is the question of \emph{identity testing} (also commonly known as one-sample testing, or goodness-of-fit), where one is given a reference distribution $\q$ and aims to test whether the unknown $\p$ is equal to this purported model $\q$, or far from it in statistical distance.

The sample and time complexity of identity testing have been thoroughly studied, and this question is now well-understood with regard to all parameters at play (see~\cite{GoldreichR00,BatuFFKRW01,Paninski08,DiakonikolasKN15a,DiakonikolasK16,Goldreich16,DiakonikolasGPP18,ValiantV17,BlaisCG19}, or the surveys~\cite{Canonne15,BalakrishnanW18}). However, at the very heart of the question's formulation lies a significant assumption: namely, that the \emph{domain} of the distributions, generally taken to be the set $[\ab]\eqdef \{1,2,\dots,\ab\}$, is the ``right'' representation of the data. In many situations, this is not the case: the observations are made with a given (often arbitrary) level of granularity, e.g., imposed by the accuracy of the measuring equipment; this may lead to falsely (over)accurate measurements, with non-significant precision in the observations. In such cases, the domain $[\ab]$ of both the model $\q$ and the distribution of the measurements $\p$ are somewhat of a red herring, and relying on it to perform identity testing may lead to a wrong answer, by introducing discrepancies where there is none. Instead, a more robust approach is to decide if $\p$ conforms to $\q$ on some suitable quantization of the data: which leads to the question, first suggested as an open question in~\cite{OpenProblem96} and introduced in this work, of \emph{testing identity up to binning}. 

In more detail, the question we consider is the following:
\begin{framed}
\noindent $\binidentityproblem(\ab,\q,\numparts,\dst)$: Given a target number of bins $\numparts$ and a reference distribution $\q$ over $[\numparts]$, a distance parameter $\dst$, and i.i.d.\ samples from an unknown distribution $\p$ over $[\ab]$, is there a partition\footnotemark{} of the domain in $\numparts$ intervals $I_1,\dots,I_\numparts$ such that $\p(I_j)=\q(j)$ for all $1\leq j\leq \numparts$, or is $\p$ at total variation distance at least $\dst$ from all distributions for which such a binning exists?
\end{framed}
\footnotetext{Throughout this paper, by \emph{partitions} we refer to ordered partitions with possibly empty subsets: i.e., a partition of $[\ab]$ in $\numparts$ sets is a sequence of pairwise disjoint sets $(I_1,\dots,I_\numparts)$ such that $\cup_{j=1}^\numparts I_j = [\ab]$, with $\abs{I_j} \geq 0$ for all $j$.}

Before proceeding further, let us discuss some aspects of this formulation (which is illustrated in~\cref{fig:example}). First of all, this formulation is intrinsically tied to discrete distributions (which aligns with the motivations laid out earlier, relating to discretization of observations): indeed, if $\p$ is continuous, then for all choices of $\q$ there always exists a suitable partition of the domain; so that $\p$ trivially satisfies the property. 
This formulation also allows us to test for some underlying structure, without caring about other irrelevant details of the distribution. For instance, setting $\q=(1/3,1/3,1/3)$ captures all distributions $\p$ with three disjoint groups of elements, separated by an arbitrary number of zero-probability elements, each with total probability mass $1/3$.

In terms of the range of parameters, we encourage the reader to keep in mind the setting where $\numparts \ll \ab$ (or even where $\numparts$ is a relatively small constant), which captures fitting an over-accurate set of measurements to a simple model. Yet, we emphasize that even the case $\numparts=\ab$ is of interest and does not collapse to identity testing. Indeed, due to our allowing empty intervals, the problem is both qualitatively and quantitatively different from identity testing, and can be interpreted as identity testing up to merging some clusters of adjacent domain elements and having zero-probability elements in the reference distributions for data points considered irrelevant (see~\cref{fig:nequalsk} for a simple illustration).\smallskip

\begin{figure}[ht!]
	\centering
	\begin{subfigure}[t]{0.40\textwidth}
		\centering
		\includegraphics[width=0.95\textwidth]{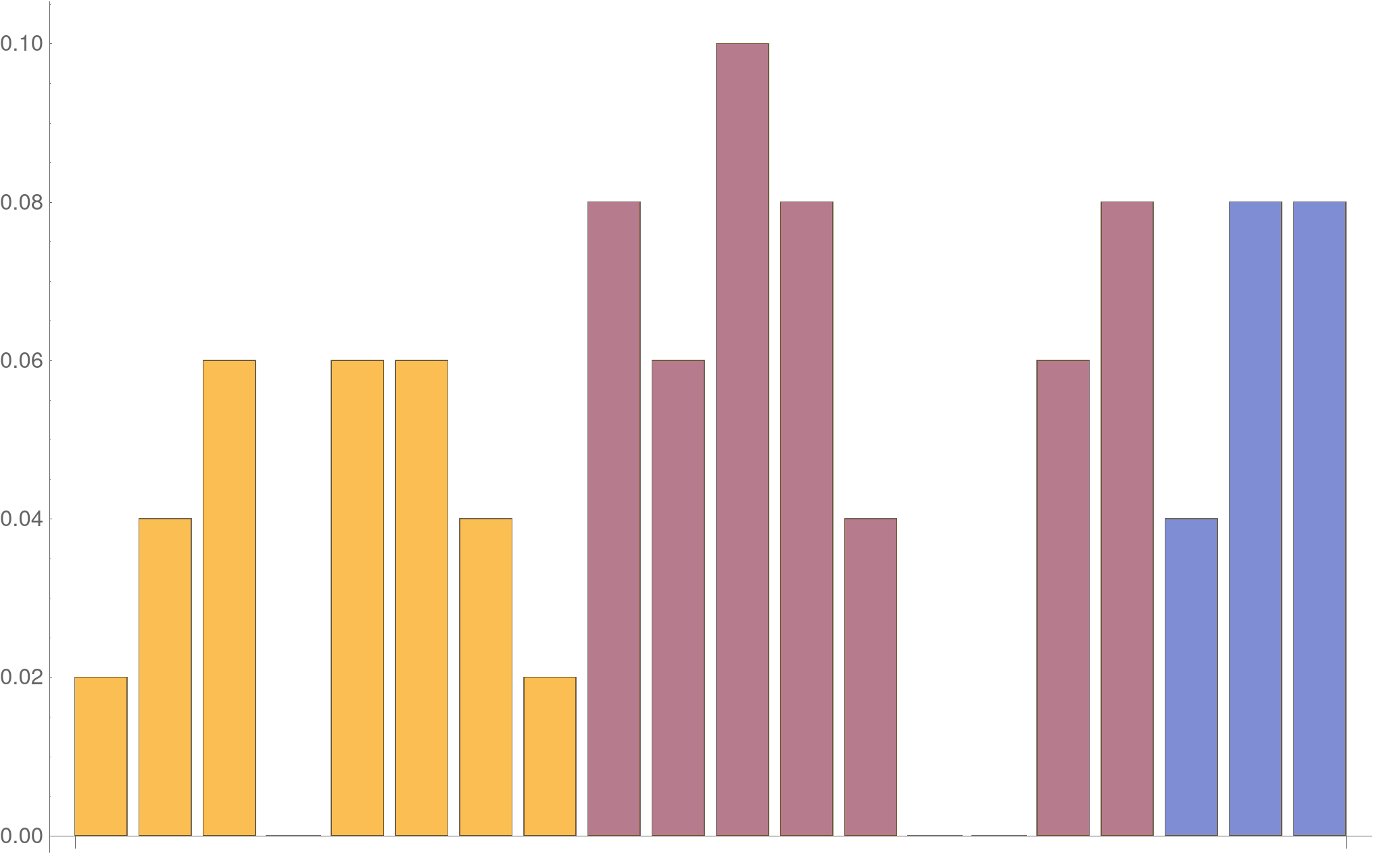}
		\caption{Distribution $\p$ on $[\ab]$}\label{fig:1a}		
	\end{subfigure}
	\qquad
	\begin{subfigure}[t]{0.40\textwidth}
		\centering
		\includegraphics[width=0.95\textwidth]{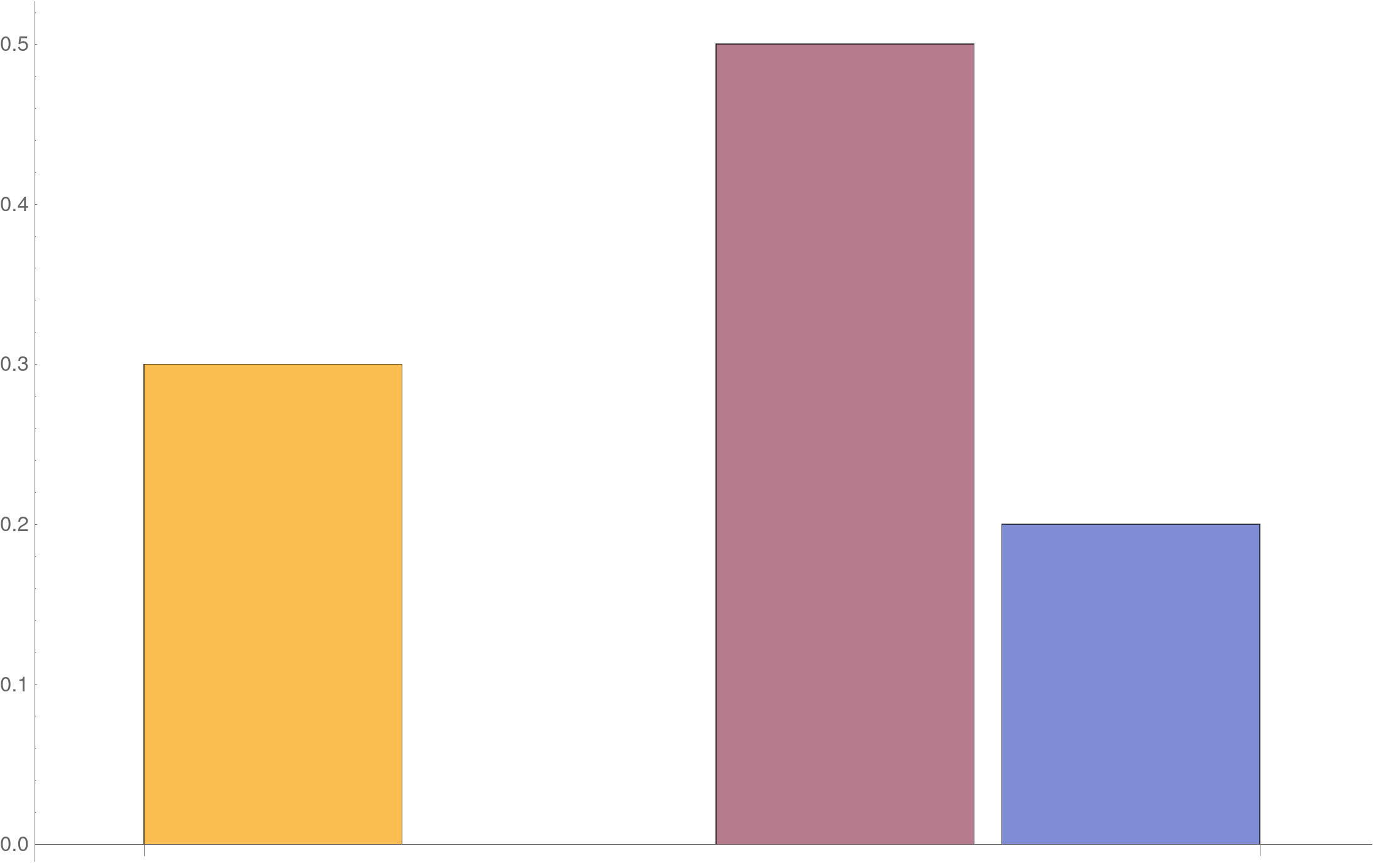}
		\caption{Distribution $\q$ on $[\numparts]$}\label{fig:1b}
	\end{subfigure}
\caption{An example, for $\ab=20$ and $\numparts=4$. Here, $\q=(3/10, 0, 1/2, 1/5)$, and $\p=\frac{1}{50}(1, 2, 3, 0, 3, 3, 2, 1, 4, 3, 5, 4, 2, 0, 0, 3, 4, 2, 4, 4)$. A possible partitioning is $I_1=\{1,2,\dots,8\}$, $I_2=\emptyset$, $I_3=\{9,\dots,17\}$, and $I_4=\{18,19,20\}$.}\label{fig:example}
\end{figure}

\begin{figure}[ht]
	\centering
  \includegraphics[width=0.5\textwidth]{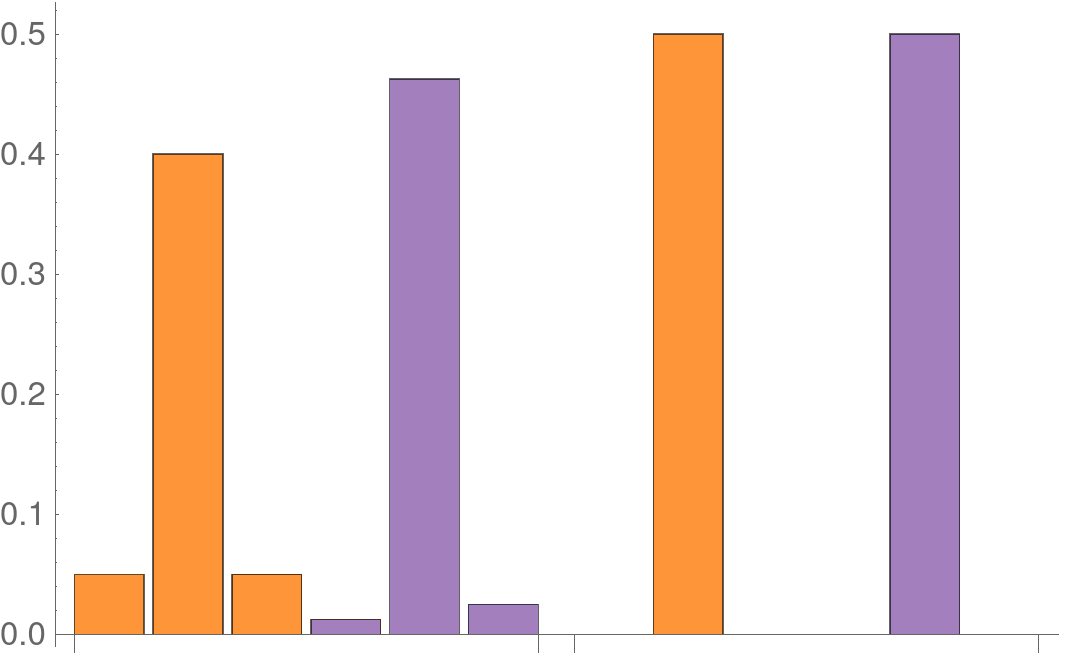}
  \caption{An example for $\ab=\numparts=6$. Here, $\q=(0,1/2,0,0,1/2,0)$ (on the right), and $\p=(1/20, 2/5, 1/20, 1/80, 37/80, 1/40)$ (on the left). The reference distribution $\q$ focuses on the two modes of the distribution (writing off the other elements as ``noise'' or spurious), while the sampled distribution $\p$ also features those elements.}\label{fig:nequalsk}
\end{figure}

We further note that one could consider variants of this problem, each with a slightly different focus. The first variant, \emph{simultaneous binning}, would feature a reference distribution $\q$ over $[\ab]$ (instead of a set of bin probabilities $\q$ on $[\numparts]$), and ask about the existence of an interval partition $I_1,\dots,I_\numparts$ such that $\p(I_j)=\q(I_j)$ for all $j$. One interesting aspect of this variant is that it generalizes non-trivially to continuous distributions (or mixtures of continuous and discrete). However, even insisting for non-empty intervals $I_j$'s, this formulation is only interesting for distributions putting significant probability mass on the first $\numparts-1$ elements.\footnote{Indeed, unless both the reference $\q$ and the unknown $\p$ put $\Omega(\dst)$ mass on $\{1,2,\dots,\numparts-1\}$, a trivial partitioning with $I_j=\{j\}$ for $1\leq j\leq \numparts-1$ and $I_\numparts=\{\numparts,\dots,\ab\}$ provides an immediate answer to the problem.} For $\numparts=o(\ab)$, and in particular constant $\ab$, this is a significant restriction. 
A second variant, which suffers the same drawback but allows more flexibility, also provides a reference distribution $\q$ over $[\ab]$, but only asks about the existence of \emph{independent binnings}: two interval partitions $I_1,\dots,I_\numparts$ and $I'_1,\dots,I'_\numparts$ such that $\p(I_j)=\q(I'_j)$ for all $j$.

We note that our upper bound (\cref{theo:upper:bound}) easily extends to these two different variants as well; and our lower bound (\cref{theo:lower:bound}) applies to the last independent binnings variant as well.

\subsection{Our results and techniques}

In this paper, we establish nearly matching upper and lower bounds on the $\binidentityproblem$ problem. Our first result is an efficient algorithm for testing the property:
\begin{theorem}
  \label{theo:upper:bound}
  For every domain size $\ab$, number of bins $1\leq \numparts\leq \ab$, and fixed reference distribution $\q$ over $[\numparts]$, there exists a (computationally efficient) algorithm for testing identity-up-to-binning to $\q$ with sample complexity $\bigO{\numparts/\dst^2}$, where $\dst\in(0,1]$ is the distance parameter.
\end{theorem}
Crucially, this upper bound is \emph{independent} of the underlying domain size $\ab$ of the unknown distribution, and only depends, linearly, on the target number of bins $\numparts$. This is to be compared to the ``standard'' identity testing problem, which is known to have sample complexity $\bigTheta{\sqrt{\ab}}$. One may wonder if the dependence on the number of bins in the above theorem is necessary, or if one could hope to achieve a \emph{sublinear} dependence on $\numparts$. Our next (and main) result shows that a near-linear dependence is unavoidable, ruling out any strongly sublinear dependence on $\numparts$:
\begin{theorem}[Informal version of~\cref{theo:lower:bound:detailed}]
  \label{theo:lower:bound}
  For every $\numparts\geq 1$, there exists $\ab = \ab(\numparts)$ and an absolute constant $\dst_0>0$ such that the following holds. There exists a reference distribution $\q$ over $[\numparts]$ such that any algorithm for testing identity-up-to-binning to $\q$ to distance $\dst_0$ must have sample complexity $\bigOmega{\numparts^{1-o(1)}}$.
\end{theorem}
\noindent Note that by standard techniques, this readily imply a lower bound of $\bigOmega{\numparts^{1-o(1)}/\dst}$, for all $\dst \in(0,\dst_0]$. We further conjecture the tight bound to be $\bigOmega{\numparts/\dst^2}$ (i.e., matching our upper bound). It is worth noting that the obvious and natural approach, a reduction from (standard) identity testing, would fall short of this goal, as it would only lead to a much weaker $\Omega(\sqrt{\numparts}/\dst^2)$ lower bound.

\paragraph{Our techniques.}

Our upper bound proceeds via the ``testing-by-learning'' paradigm.  Specifically, given a reference distribution $\q$ over $[\numparts]$, and sample access to an unknown distribution $\p$ over $[\ab]$, we learn a hypothesis distribution $\hat{\p}$ such that, informally, for every binning of $[\ab]$ into $\numparts$ intervals, the distributions $\p$ and $\hat{\p}$ are close.  (More formally, the distribution $\hat{\p}$ to close to $\p$ in $\akdist{\numparts}$-distance, which we define later.) Since we have exact access to $\hat{\p}$, we can enumeratively test every potential binning.  Although this does not save on the query complexity, we are able to use dynamic programming to find the best binning in time polynomial in $\ab$ and $\numparts$. 

The proof of our lower bound is significantly more involved and proceeds in three stages, which we outline here.  We first show that, if $\ab$ is sufficiently large as a function of $\numparts$, then we may assume that the algorithm only looks at the ordering of the samples received from $[\ab]$, instead of at their actual ``names'' (we refer to the ordering as the \emph{ordered fingerprint} of the sample).  Our bound uses a Ramsey theory argument very similar to that of~\cite{DiakonikolasKN15}.

We then show, for every integer $m$, the existence of two different distributions $\p'$ and $\q'$ over $[\poly(m) \cdot 2^m]$ such that (i) $\p'$ and $\q'$ are far in total variation distance, and (ii) given the ordered fingerprint of $m$ samples generated by one of $\p'$ and $\q'$, it is information-theoretically impossible to determine which of $\p'$ and $\q'$ generated the ordered fingerprint.  There are $2^m$ such ordered fingerprints, one for every composition of $m$, so it is unclear if the support size of the distributions in this claim can be significantly reduced.  While this is an independently interesting problem, any progress on this problem would only directly result in a decrease of the $o(1)$ exponent in our main theorem.

Finally, using the above pair of distributions as ``buckets'' of two distributions $\p$ and $\q$, respectively, we show that any algorithm distinguishes $\p$ and $\q$ and only considers ordered fingerprints must have access to at least $m+1$ samples from at least one of the (assuming $m$ is independent of $\numparts$) $\Omega(\numparts)$ buckets.  By collision bounds for the generalized birthday problem (see also~\cite{SuzukiTKT06} for more formal statements), with high probability we need a sample size of $\Omega(k^{1-1/(m+1)})$ to ensure that the algorithm succeeds.

\subsection{Future directions}

We conclude with a few research directions we deem particularly promising. A first natural question is to understand the analogue of the $\binidentityproblem$ problem for \emph{closeness} testing, that is, when both $\p$ (over $[\ab]$) and $\q$ (over $[\numparts]$) are unknown and available only through i.i.d. samples. This would in particular capture situations where data is collected from two different sources (or sensors), each using a different discretization scheme or with different precision, and one aims at deciding whether the underlying distributon is the same.

Another avenue is to consider the \emph{tolerant} testing version of the problem, that is, to decide whether $\p$ is \emph{close} to a binning of $\q$ (versus far from any such binning). It is known that tolerant identity testing has significantly larger (namely, nearly linear in the domain size) sample complexity than identity testing; whether this is still the case for their identity-up-to-binning analogues is an intriguing question.

Lastly, it would be interesting to generalize the question to other partially-ordered domains (e.g., the hypergrid $[\ab]^d$, or the Boolean hypercube $\{0,1\}^\ab$), for the suitable notion of ``interval'' in these posets.

\section{Preliminaries}
All throughout the paper, we write $\log$ for the binary logarithm. Given integers $\ab$ and $\numparts$, we write $\mathcal{J}_{\ab,\numparts}$ for the set of all $\binom{\ab+\numparts-1}{\ab}$ partitions of $[\ab]$ in $\numparts$ consecutive (and possibly empty) pairwise disjoint intervals. Recall that the total variation distance between two distributions $\p,\p'$ over $[\ab]$ is $\totalvardist{\p_1}{\p_2} = \sup_{S\subseteq [\ab]}(\p(S)-\p'(S))$, and is equal to $\frac{1}{2}\sum_{i=1}^\ab \dabs{\p(i)-\p'(i)}$, half the $\lp[1]$ distance between their probability mass functions.

The $\binidentityproblem(\q,\numparts,\dst)$ problem we consider in this paper is then formally defined as follows. Given a reference distribution $\q$ on $[\numparts]$, we consider the property
\begin{equation}
\cP_{\q}\eqdef\setOfSuchThat{ \p'\in\distribs{\ab} }{ \exists (I_1,\dots,I_\numparts) \in \mathcal{J}_{\ab,\numparts}\,,\;\; \max_{j\in[\numparts]} \abs{ \p'(I_j) - \q(j) } = 0 }
\end{equation}
The question is then, given a distance parameter $\dst\in(0,1]$, to distinguish (i)~$\p\in\cP_{\q}$ from (ii)~$\totalvardist{\p}{\cP_{\q}} > \dst$, where $\totalvardist{\p}{\cP_{\q}} = \min_{\p'\in\cP_{\q}}\totalvardist{\p}{\p'}$. 
We will rely extensively, for our lower bounds, on the below fact, which allows us to lower bound $\totalvardist{\p}{\cP_{\q}}$ by a more actionable quantity.\begin{fact}\label{fact:dist}
For any distributions $\p$ and $\q$ (over $[\ab]$ and $[\numparts]$, respectively) we have
\[
\totalvardist{\p}{\cP_{\q}} \geq \frac{1}{2}\dist{\p}{\q} \eqdef \min_f \sum_{j=1}^\numparts \big| \q(j) - \sum_{t: f(t) = j} \p(t) \big|,
\]
where the minimum is taken over all nondecreasing functions $f\colon [\ab]\to[\numparts]$.
\end{fact}
\begin{proof}
Fix $\p,\q$ as in the statement. One can think of such a nondecreasing $f$ as defining a decomposition of $[\ab]$ into $\numparts$ disjoint (and possible empty) intervals $I_1,\ldots,I_\numparts$, where $I_j = f^{-1}(j)$. In particular, let $\p^\ast\in\cP_{\q}$ such that $\totalvardist{\p}{\p^\ast}=\totalvardist{\p}{\cP_{\q}}$, and $I^\ast_1,\dots,I^\ast_\numparts$ the corresponding partition for $\p^\ast$. Define $f^\ast\colon [\ab]\to[\numparts]$ from this partition by setting $f(t) \eqdef \sum_{j=1}^\numparts j\indicSet{I^\ast_j}(t)$: $f^\ast$ is then non-decreasing, and
\[
    \sum_{j=1}^\numparts \big| \q(j) - \sum_{t: f^\ast(t) = j} \p(t) \big| 
     = \sum_{j=1}^\numparts \big| \q(j) - \p(I^\ast_j) \big|
     = \sum_{j=1}^\numparts \big| \p^\ast(I^\ast_j) - \p(I^\ast_j) \big|
     \leq \sum_{i=1}^\ab \big| \p^\ast(i) - \p(i) \big|
\]
where the last equality is due to the fact that $\p^\ast(I^\ast_j)=\q(j)$ for all $j\in[\numparts]$ (since $\p^\ast\in\cP_{\q}$), and last inequality is the triangle inequality. Since  the RHS equals $2\totalvardist{\p}{\p^\ast}$, 
we get $\min_f \sum_{j=1}^\numparts | \q(j) - \sum_{t: f(t) = j} \p(t) | \leq 2\totalvardist{\p}{\cP_{\q}}$.
\end{proof}
 
\section{Algorithms}

In this section, we prove~\cref{theo:upper:bound}, restated below:
\begin{theorem}
  \label{theo:upper:bound:restated}
  For any fixed reference distribution $\q$ over $[\ab]$ and number of bins $1\leq \numparts\leq \ab$, there exists a (computationally efficient) algorithm for $\binidentityproblem(\ab,\q,\numparts,\dst)$ with sample complexity $\bigO{\numparts/\dst^2}$.\end{theorem}
\begin{proof}
  We will use the following notion of \emph{$\akdist{\ell}$-distance}, which interpolates between Kolmogorov distance ($\ell=2$) and (twice the) total variation ($\ell=\ab$)~\cite{DevroyeLugosi01,ChanDSS12,DiakonikolasKN15a,DiakonikolasKN15}: for any two distributions $\p',\q'$ over $[\ab]$,
  \begin{equation}
      \norm{\p'-\q'}_{\akdist{\ell}} \eqdef \max_{ (I_1,\dots,I_\ell) \in \mathcal{J}_{\ab,\ell} } \sum_{j=1}^\ell \abs{ \p'(I_j) - \q'(I_j) }
\end{equation}
One can check that this defines a \textit{bona fide} norm. Further, it is known, from the VC inequality, that $\bigO{\ell/\dst^2}$ samples are sufficient to learn any distribution in $\akdist{\ell}$-distance $\dst$ (and with failure probability $1/6$)~\cite{DevroyeLugosi01,ChanDSS12}, and, further, that the empirical estimator achieves this bound. Using this, the algorithm is as follows:
\begin{enumerate}
  \item Learn the unknown $\p$ to $\akdist{\numparts}$-distance $\dst$, with failure probability $1/6$. Call the result (the empirical estimator, which is a distribution over $[\ab]$) $\hat{\p}$.
  \item Compute the minimum $\Delta_{\q}(\mathcal{I},\hat{\p}) \eqdef  \sum_{j=1}^\numparts \abs{ \hat{\p}(I_j) - \q(j) }$ over all partitions $(I_1,\dots,I_\numparts) \in \mathcal{J}_{\ab,\numparts}$ in intervals such that $\abs{I_j} > 0$ whenever $\q(j) > 0$.
  \item If there exists $\mathcal{I}$ such that $\Delta_{\q}(\mathcal{I},\hat{\p}) \leq \dst$, return \accept; otherwise, return \reject.
\end{enumerate}
Note that only the first step requires samples from $\p$, so that the sample complexity is indeed $\bigO{\numparts/\dst^2}$. The second step is purely computational, and can be implemented in time $\poly(\ab,\numparts)$ via a simple dynamic programming approach.

We now argue correctness. With probability at least $5/6$, the first step produced a correct $\hat{\p}$, i.e., one that is $(\dst/4)$-close to $\p$ in $\akdist{\numparts}$-distance; we hereafter assume this holds.

\myparagraph{Completeness.} Suppose $\p\in\cP_{\q}$, and let $\mathcal{I}^\ast=(I^\ast_1,\dots,I^\ast_\numparts)\in \mathcal{J}_{\ab,\numparts}$ be any partition witnessing it. Note that this partition then satisfies $\dabs{I^\ast_j} > 0$ for all $j$ such that $\q(j) > 0$, as otherwise the distance after binning is positive. Then,
\[
    \Delta_{\q}(\mathcal{I}^\ast,\hat{\p})
    = \sum_{j=1}^\numparts \abs{ \hat{\p}(I^\ast_j) - \q(j) }
    \leq \sum_{j=1}^\numparts \abs{ \hat{\p}(I^\ast_j) -\p(I^\ast_j) } + \sum_{j=1}^\numparts \abs{ \p(I^\ast_j)  - \q(j) }
    \leq \norm{\hat{\p}-\p}_{\akdist{\numparts}}
    \leq \dst
\]
the second-to-last inequality by definition of $\akdist{\numparts}$-distance and the fact that $\sum_{j=1}^\numparts \dabs{ \p(I^\ast_j)  - \q(j) } = 0$. Therefore, the algorithm will find a good partition and output \accept.

\myparagraph{Soundness.} Suppose now by contrapositive that the algorithm outputs $\accept$. This means it found, in step 2, some partition $\mathcal{I}^\ast=(I^\ast_1,\dots,I^\ast_\numparts)\in \mathcal{J}_{\ab,\numparts}$ such that $\Delta_{\q}(\mathcal{I}^\ast,\hat{\p}) \leq \dst/4$. But then,
\[
    \sum_{j=1}^\numparts \abs{ \p(I^\ast_j)  - \q(j) } = \sum_{j=1}^\numparts \abs{ \p(I^\ast_j) - \hat{\p}(I^\ast_j) } + \sum_{j=1}^\numparts\abs{ \hat{\p}(I^\ast_j)  - \q(j) }
    \leq \norm{\hat{\p}-\p}_{\akdist{\numparts}} + \Delta_{\q}(\mathcal{I}^\ast,\hat{\p})
    \leq 2\dst
\]
the second-to-last inequality again by definition of $\akdist{\numparts}$-distance. We claim that this implies that $\p$ is $\dst$-close (in total variation distance) to some $\p^\ast\in\cP_{\q}$. Indeed, we can build $\p^\ast$ in a greedy fashion from $\p$: as long as there exist $I^\ast_{j_1}$, $I^\ast_{j_2}$ such that $\p(I^\ast_{j_1}) > \q(j_1)$ but $\p(I^\ast_{j_2}) < \q(j_2)$, we move $\delta \eqdef \min(\dabs{\p(I^\ast_{j_1}) - \q(j_1)},\dabs{\p(I^\ast_{j_2}) - \q(j_2)})$ probability mass from (arbitrary) elements of $I^\ast_{j_1}$ to an arbitrary element of $I^\ast_{j_2}$. Here, we used our condition on the partition found in step 2, which ensures $I^\ast_{j_2}$ is non-empty. This incurs total variation $\delta$ (from $\p$), and reduces $\sum_{j=1}^\numparts \dabs{ \p(I^\ast_j)  - \q(j) }$ by $2\delta$. Repeating until it is no longer possible, we obtain the claimed $\p^\ast$, and therefore that $\totalvardist{\p}{\cP_{\q}} \leq \dst$.
\end{proof}
\begin{remark}
As mentioned in the introduction, it is straightforward to adapt the algorithm and argument to either the ``simultaneous binning'' and ``independent binnings'' variants of the problem, leading to the same sample complexity upper bound (and time complexity) for these related questions.
\end{remark}

\section{Lower Bounds}

In this section, we prove our main theorem (\cref{theo:lower:bound}), restated below:
\begin{theorem}
  \label{theo:lower:bound:detailed}
  For every $\numparts\geq 1$, there exists $\ab = \ab(\numparts)$ such that the following holds. There exists a reference distribution $\q$ over $[\numparts]$ such that any algorithm for $\binidentityproblem(\ab,\q,\numparts,\dst)$ must have sample complexity $\frac{\numparts}{2^{O(\sqrt{\log\numparts})}}$.
\end{theorem}
\noindent At a high level, the proof of~\cref{theo:lower:bound:detailed} consists of three parts:
\begin{enumerate}[(1)]
	\item\label{item:proof:overview:1} First, we show that, for every $\ab$, there is an $N(\ab)$ such that, if there is a algorithm that succeeds in the data binning problem over $[N(\ab)]$ with $\ns$ samples, then there is an algorithm that succeeds in the data binning problem over $[\ab]$ with $\ns$ samples, such that this latter algorithm only considers the ``ordered fingerprint'' of the sample. This approach, analogous to that of~\cite{DiakonikolasKN15}, uses Ramsey theory, and enables us to restrict ourselves to (simpler to analyze) order-based algorithms. (\cref{prop:reduction:ordered:fingerprints:only})
	\item\label{item:proof:overview:2} Second, we show for every constant $m$, there exist $\blksz(m)$ and two distributions $\psmall$ and $\qsmall$ on $[\blksz]$ that (i)~are ``far from being cyclic shifts'' of one another (which will allow us to argue about distance to being a binning later), and (ii)~cannot be distinguished from ordered fingerprints from $\ns$ samples. The function $\blksz(m)$ is exponential in $m$. (\cref{lem:two-dist:same-moments})
	\item\label{item:proof:overview:3} Third, we consider the distributions $\p$ and $\q$ on $[\numparts]$ constructed such that the distribution restricted to each block of $\blksz(m)$ elements is $\psmall$ and $\qsmall$, respectively, in the natural way.  We show that (i)~$\pbig$ is ``far'' from any binning of $\qbig$, and (ii)~distinguishing these two distributions with ordered fingerprints requires seeing $m + 1$ samples from one particular block.  By folklore collision bounds, this implies a $\Omega(\numparts^{1-1/(m+1)})$ sample lower bound for distinguishing $\pbig$ and $\qbig$ from ordered fingerprints. (\cref{cor:lower:bound:kvk})
\end{enumerate}
Combining the three (applying~\ref{item:proof:overview:1} with $\ab$ set to $\numparts$, and $m$ in \ref{item:proof:overview:3}  set to $\sqrt{\log \numparts}$) then yields the theorem.

\subsection{From samples to ordered fingerprints}
The first step of our reduction will consist in ``hiding'' information from the algorithms. Roughly speaking, for the sake of our lower bound analysis we would like to argue that, without loss of generality, any testing algorithm can be assumed to only be given the \emph{order relations} between the samples from the unknown distribution, instead of the samples themselves. For instance, instead of seeing four samples $12,7,98,7$, we wish to restrict our analysis to algorithms which only see that $x_{(1)} = x_{(2)} < x_{(3)} < x_{(4)}$. Since we are drawing i.i.d. samples from the distribution, the order that the samples come in does not matter.

However, this simplifying assumption is not actually without loss of generality, and does not quite hold as stated: an algorithm \emph{can} sometimes infer more information from the values of the samples from $\p$ than from their ordered relations alone. We will prove a weaker statement, sufficient for our purposes; in order to do so, we start by introducing some notions formalizing the aforementioned ``order relations.''

\begin{definition}[Ordered fingerprints]
The \emph{ordered fingerprint} of a sequence of $\ns$ values in $[\ab]$ is an ordered frequency vector of the elements occurring in the sequence, with labels removed. Formally, for each element $i \in [\ab]$, let $\bF_i$ be the number of times element $i$ appears, and let $j_1<j_2<\ldots<j_t$ be the indices such that $\bF_{j_i}$ is positive (so that $t$ is the number of distinct elements in the sequence).  The ordered fingerprint is then the ordered collection of positive integers $(\bF_{j_1},\bF_{j_2},\ldots,\bF_{j_t})$ such that $\sum_{i=1}^t \bF_{j_i} = \ns$.
\end{definition}

To later argue about indistinguishability of our instances, we will also rely on the below notion of $\ns$-way moments (induced by an ordered fingerprint):
\begin{definition}[$s$-way moments]
Given a distribution $\p$ over $[\ab]$ and an ordered tuple of positive integers $F = (F_1,F_2,\ldots,F_t)$ such that $\sum_{i=1}^t F_i = \ns$, the probability that the ordered fingerprint on $\ns$ samples is $(F_1,F_2,\ldots,F_t)$ is given by
\begin{equation}\label{eq:m:way:moments}
\p^F \eqdef \dbinom{\ns}{F_1,F_2,\ldots,F_t} \sum_{1 \leq i_1 < i_2 < \cdots < i_t \leq \ab} \prod_{j=1}^t \p(i_j)^{F_{j}}
\end{equation}
We call such an expression an \emph{$\ns$-way moment of $\p$}.  Note that the $\ns$-way moments of $\p$ completely determine the distribution of the ordered fingerprint of $\ns$ random samples drawn from $\p$.   Thus, for two distributions $\p$ and $\q$, if $\p^F = \q^F$ for all $F$ such that $|F| = \ns$, then $\p$ and $\q$ cannot be distinguished from ordered fingerprints on $\ns$ samples.\footnote{Indeed, for every $s$, the $(s-1)$-way moments are linear combinations of $s$-way moments; we omit the details.}
\end{definition}
We now state and prove the following lemma, which captures the intuition discussed above and will be the first component of our lower bound:
\begin{proposition}
  \label{prop:reduction:ordered:fingerprints:only}
For every $\ab\geq 1$ and $\ns\geq 1$, there exists $N\geq 1$ such that the following holds for every $\numparts\geq 1$, reference distribution $\q$ over $[\numparts]$, and $\dst\in(0,1]$. If there exists an algorithm $\Algo_N$ that, for every distribution $\p'$ over $[N]$, solves the problem $\binidentityproblem(N,\q,\numparts,\dst)$ with $\ns$ samples from $\p'$, then there exists an algorithm $\Algo_\ab$ that, for every distribution $\p$ over $[\ab]$,  solves the problem $\binidentityproblem(\ab,\q,\numparts,\dst)$ with $\ns$ samples from $\p$. Moreover, $\Algo_\ab$ only considers the ordered fingerprint of the $\ns$ samples.
\end{proposition}
\begin{proof}
The argument is similar to that of Diakonikolas, Kane, and Nikishkin~\cite{DiakonikolasKN15}, and relies on a result of Conlon, Fox, and Sudakov:
\begin{lemma}[{\cite{ConlonFS10}}]
\label{lem:ramsey}  Given a set $S$ and an integer $t$, let $\binom{S}{t}$ denote the set of subsets of cardinality $t$.  For all positive integers, $a$, $b$, and $c$, there exists a positive integer $N$ so that for any function $\chi\colon\binom{[N]}{a} \to [b]$, there exists $S \subseteq [N]$ with $|S| = c$ such that $\chi$ is constant on $\binom{S}{a}$.
\end{lemma}
As in the proof of~\cite[Theorem~13]{DiakonikolasKN15}, given $\ab,\ns$, we will invoke~\cref{lem:ramsey} to obtain a new domain size $N$, in a way detailed below: loosely speaking, $a$ corresponds to the number of samples $\ns$, $b$ to the number of possible decision (binary) functions from a given set of $\ns$ samples which only depend on the ordered fingerprint of those samples, and $c$ to the target domain size $\ns$. We will use this to, given $\Algo_N$ (which induces a mapping $\chi$ from sets of $\ns$ samples to such order-based decision functions), find a set $S \subseteq [N]$ of size $\ab$, and a monotone function $f\colon [\ab] \to [N]$ such that $f([\ab])=S$. The algorithm $\Algo_\ns$ then runs the promised algorithm $\Algo_N$ on the distribution $f(\p)$ obtained by applying $f$ to the $\ns$ samples. Since $f$ is order-preserving, the distances will not change.  Since we can assume that $\Algo_N$ is deterministic once we know which distinct samples we get from $[\ab]$, the fact that the induced $\chi$ is constant on $\binom{f([\ab])}{\ns}$ guarantees the output will be a function of the ordered fingerprint only.

However, we have an issue that~\cite{DiakonikolasKN15} does not.  When we draw a total of $\ns$ samples, we likely will not get $\ns$ distinct samples. \cite{DiakonikolasKN15} gets around this issue by ``dividing'' all elements of $[\ab]$ into sub-elements, and upon seeing an element in the sample, assign it to a uniformly random one of its sub-elements.  However, our main parameter of interest $\numparts$ differs from theirs: if we applied this procedure to our problem, the bounds obtained would deteriorate.  That is, while the $\mathcal{A}_\numparts$-distance (the focus of~\cite{DiakonikolasKN15}) would not change, the coarsening distance we consider here would.

We handle this issue as follows.  Given an integer $\ns$, we will color all nonempty subsets of $[N]$ with \emph{at most} $\ns$ elements.  This color associates to a set of $\ns$ samples the function that $\Algo_N$ uses to \accept or \reject given the ordered fingerprint that accompanies these samples.  Now, we apply~\cref{lem:ramsey} repeatedly.  In our first application, we find a subset $S_\ns \subseteq [N]$ such that, conditioning on distinct samples, $\Algo_N$ has consistent behavior.\footnote{In this case, the ordered fingerprint is the all ones vector of length $\ns$, so the functions are not very interesting.}{}  Now we remove all subsets that contain elements outside of $S_\ns$, and we apply~\cref{lem:ramsey} again to find a set $S_{\ns-1} \subseteq S_s$ such that $\Algo_N$'s behavior only depends on the ordered fingerprint given that in the $\ns$ samples, there are at least $\ns-1$ distinct values seen.
Continuing in this fashion, after applying~\cref{lem:ramsey} a total of $\ns$ times, we arrive at a set $S_1$ (with $S_1\subseteq S_2\subseteq \dots\subseteq S_\ns$) such that, conditioned on samples coming only from $S_1$, $\Algo_N$ depends only on the ordered fingerprint of the samples seen.
In every application of the lemma, we have $a \leq \ns$, and since there are $2^\ns$ possible ordered fingerprints given $\ns$ samples, we have $b \leq 2^{2^\ns}$.  We set $c = \ab$ in our final application.
\end{proof}

\subsection{Sequences far from being shifts of each other}
As a building block from our lower bound construction, we will require the existence of (non-negative) sequences which are ``far'' from each other, in that any circular shift of the first remains far, in Hamming distance, from the second.

\begin{definition}[Partial shifts]
Given an alphabet $\Sigma$ and two strings $x$ and $y$ in $\Sigma^\ab$, we say that $y$ is an \emph{$r$-partial cyclic shift of $x$} if there is an integer $\ell$ and a nondecreasing function $f\colon [\ab] \to [\ab]$ such that, for at least $r$ indices $i$ of $[\ab]$, $x_i = y_{(f(i)+\ell) \bmod n}$.  That is, thinking of the symbols of $y$ placed in a circle, $y$ contains a substring $r$ of $x$ (without wrapping around more than once).

For two distributions $\p$ and $\q$ over $[\ab]$, we then say that \emph{$\q$ is an $r$-partial cyclic shift of $\p$} if the string $\q(1)\q(2)\cdots\q(\ab)$ is an $r$-partial cyclic shift of $\p(1)\p(2)\cdots\p(\ab)$.
\end{definition}
We observe that finding a witness to $x$ and $y$ being $r$-partial cyclic shifts of each other is equivalent to finding the \emph{cyclic longest common subsequence} between $x$ and $y$, which is a problem with applications in DNA sequencing~\cite{Nguyen12}. 
With these notions in hand, we are ready to prove the key technical lemma underlying part~\ref{item:proof:overview:2} of our argument:
\begin{lemma}
\label{lem:two-dist:same-moments}
	Let $m > 4$ be a positive integer, and $\blksz = 5m^2 2^m$. Then there exist two distributions $\psmall$ and $\qsmall$ over $[\blksz]$ such that (i)~$\psmall$ and $\qsmall$ are not $(99b/100)$-partial cyclic shifts of each other, and (ii)~$\psmall$ and $\qsmall$ agree on all $m$-way moments.
\end{lemma}
\begin{proof}
	Let $\p$ be a probability distribution over $[\blksz]$ such that its probability mass function takes the values $2/(\frac52b)$ and $3/(\frac52\blksz)$ each on $\blksz/2$ of the elements of $[\blksz]$.  Each $m$-way moment $\p^F$ can be expressed as a homogeneous polynomial in the $\blksz$ values $\p(i)$ of total degree $m$, such that each coefficient is $0$ or $K_F \eqdef \binom{s}{F_1,F_2,\ldots,F_t}$; note that $K_F$ is independent of $\p$.  The number of terms of such a polynomial is at most $\binom{\blksz+m-1}{m} \leq (\blksz+m-1)^m \leq (2\blksz)^m$.  Thus, every $m$-way moment of $\p$ evaluates to an expression of the form $qK_F/(\frac52\blksz)^m$, where $q$ is a integer such that $0 \leq q \leq (6\blksz)^m$.  Since there are $2^m$ $m$-way moments of $\p$, there are at most $(6\blksz)^{m2^m}$ possible values of the $m$-way moments of $\p$.
		
	For the claim about cyclic shifts, we use a counting argument.  We will show that for any string $x \in \{2,3\}^\blksz$, there are at most $2^{\blksz/5}$ other strings that are $(99\blksz/100)$-partial cyclic shifts of $x$.  Indeed, any $(99\blksz/100)$-partial cyclic shift $y$ of $x$ can be constructed by the following process:  first, select $99\blksz/100$ positions of $x$ that will be present in $y$--these positions will form a witness.  Next, select $99\blksz/100$ positions of $y$ that these bits of $x$ will appear in, and a cyclic shift for these bits.  Finally, select any values in $\{2,3\}$ for the remaining $\blksz/100$ positions of $y$.  Thus, the number of possible $y$'s is at most
	\[
	\binom{\blksz}{\blksz/100}^2 (99\blksz/100) 2^{\blksz/100} \leq 2^{2h(1/100)} 2^{\blksz/50} \leq 2^{\blksz/5} 
	\]
	where we have used the bound $\binom{n}{\alpha n} \leq 2^{nh(\alpha)}$, where $h(\alpha) = \alpha \log \alpha + (1 - \alpha) \log (1 - \alpha)$ is the binary entropy.
	Since there are a total of $\binom{\blksz}{\blksz/2} \geq 2^\blksz/\blksz$ strings of length $\blksz$ with an equal number of $2$'s and $3$'s, then for every $\blksz$, we can find a set $S$ of $2^{4\blksz/5}/\blksz$ strings of length $\blksz$ such that no two of them are $(99\blksz/100)$-partial cyclic shifts of each other.  We map strings over $\{2,3\}^\blksz$ to distributions over $[\blksz]$ in a natural way:  given $x \in S$, we define the corresponding distribution $\p$ such that $\p(i) = x_i/(\frac52\blksz)$.  Thus, we have $2^{4\blksz/5}/\blksz$ distributions over $[\blksz]$, no two of which are $(99\blksz/100)$-partial cyclic shifts of each other. 
	Setting $\blksz = 5m^2 2^m$, we then have
	\begin{align*}
	\frac{1}{\blksz} 2^{4\blksz/5} 
	&= \frac{2^{4m^2 2^m}}{5m^22^m} 
	\geq 2^{3m^22^m} 
	=(2^{3m})^{m2^m} 
	> (30m^2 2^m)^{m2^m} 
	= (6\blksz)^{m2^m}.
	\end{align*}
	Thus, by the pigeonhole principle there exist two distributions $\psmall$ and $\qsmall$ over $[\blksz]$ that are not $(99\blksz/100)$-partial cyclic shifts of each other, yet such that $\psmall$ and $\qsmall$ agree on all $m$-way moments.
\end{proof} 

\subsection{Constructing the hard-to-distinguish instances over $[\ab]$}
For a fixed integer $m > 1$ (to be determined later) and sufficiently large $\numparts$, let $\psmall$ and $\qsmall$ be the two distributions over $[\blksz]$ promised by~\cref{lem:two-dist:same-moments}, where $\blksz = 5m^2 2^m$; and let $\numblk \eqdef \numparts/\blksz$.  We will define $\pbig$ and $\qbig$ to be distributions over $[\blksz \numblk]$ in the following way:  We partition $[\blksz\numblk]$ into blocks $B_1,B_2,\ldots,B_{\numblk}$, where $B_j=\{\blksz(j-1)+1,\blksz(j-1)+2,\ldots,\blksz j\}$.  Now $\pbig$ (resp. $\qbig$) is the distribution resulting from the following process: 
\begin{itemize}
\item Pick a uniformly random element $i \in [\numblk]$ and a random sample $j$ from $\psmall$ (resp. $\qsmall$).
\item Output, as resulting sample from $\pbig$ (resp. $\qbig$), the element $(i-1)\numblk+j$.
\end{itemize}
Since $\psmall(i)$ and $\qsmall(i)$ are in $\left\{4/(5\blksz),6/(5\blksz)\right\}$ for $i \in [\blksz]$, we have that $\pbig(i)$ and $\qbig(i)$ are always in $\left\{4/(5\blksz \numblk),6/(5\blksz \numblk)\right\}$ for all $i \in [\blksz \numblk]$.

We will prove that $\totalvardist{\pbig}{\cP_{\qbig}}$ is large.  To facilitate this, recalling~\cref{fact:dist} we will lower bound the distance of $\pbig$ to a binning of $\qbig$ in the following way:
\begin{equation}\label{eq:dist}
	\totalvardist{\pbig}{\cP_{\qbig}} \geq \frac{1}{2} \dist{\pbig}{\qbig} = \frac{1}{2}\min_f \sum_{i=1}^{\blksz \numblk} \big| \qbig(i) - \sum_{t: f(t) = i} \pbig(t) \big|,
\end{equation}
where the minimum is taken over all nondecreasing functions $f\colon [\blksz \numblk] \to [\blksz \numblk]$.  To use the language of a partition\footnote{Recall that the sets in the partition are allowed to be empty.}, the interval $I_i$ is simply $f^{-1}(i)$. 
We show the following:
\begin{claim}
  \label{claim:dist:eps0}
	There exists $\dst_0>0$ such that, for $\pbig$ and $\qbig$ as described above, $\totalvardist{\pbig}{\cP_{\qbig}} \geq \dst_0$. (Moreover, one can take $\dst_0 = 10^{-7}$.)
\end{claim}
\begin{proof}
We will show that $\dist{\pbig}{\qbig} \geq 1/(5\cdot 10^6)$. We first rewrite the distance as 
$
\dist{\pbig}{\qbig} = \min_f \sum_{j=1}^{\numblk} d(B_j),
$ 	
where, using the same blocks as described earlier, we define 
\[
d(B_j) \eqdef \sum_{i \in B_j} \big| \qbig(i) - \sum_{t: f(t) = i} \pbig(t) \big|
\]
 (suppressing its dependence on $\pbig,\qbig$, and $f$).  We will call a block $B_j$ \emph{good} if $d(B_j) \leq 1/(2500\numblk)$.  Note that $f$ need not be a function that maps elements of $B_j$ to other elements of $B_j$.  Since $\pbig(i)$ and $\qbig(i)$ are always in $\left\{\frac4{5\blksz \numblk},\frac6{5\blksz \numblk}\right\}$, any element $i \in B_j$ such that $|f^{-1}(i)| \neq 1$ will contribute at least $\frac2{5\blksz \numblk}$ to $d(B_j)$.  Thus, we call an element $i$ \emph{good} if there is a unique $t$ such that $f(t) = i$ and $\qbig(i) = \pbig(t)$.  Since each element that is not good contributes $\frac2{5\blksz \numblk}$ to $d(B_j)$, it follows that, in a good block, at most $1/(2500 \numblk)/(\frac2{5\blksz \numblk}) = \blksz/1000$ elements are not good, and therefore a good block must contain at least $999\blksz/1000$ good elements.

Let $B_j$ be a good block: it must be the case that for at least $999\blksz/1000$ values of $g \in B_j$, there exists a unique $t$ such that $f(t) = g$.  Let $g_1 < g_2 < \ldots < g_v$ be the good elements in $B_j$ (where $v \geq 999\blksz/1000$), and let $t_1,t_2,\ldots,t_v$ be such that $f(t_i) = g_i$ for $i=1,2,\ldots,v$.  By assumption, the distributions $\p$ and $\q$ are not $(99\blksz/100)$-partial cyclic shifts of each other.  Thus, for every $i'$ and $i''$ where $i' < i''$, if $t_{i''} - t_{i'} < \blksz$, then $i'' - i' \leq 99\blksz/100$, or else the good elements of $B_j$ and their preimages under $f$ would be a witness to $\psmall$ and $\qsmall$ being $(99b/100)$-partial cyclic shifts of one another.  It follows that
\[
t_v - t_1 \geq \blksz + (999\blksz/1000 - 99\blksz/100) \geq (1 + 1/1000)\blksz.
\]

That is, every good block ``uses up'' a bit more than $\blksz$ elements of the domain of $f$. Since $f\colon [\blksz \numblk] \to [\blksz \numblk]$ must be a nondecreasing function, we can have at most $\dfrac{\blksz \numblk}{(1+1/1000) \blksz} \leq (1-1/2000)\numblk$ good blocks.  Thus, at least $\numblk/2000$ of the blocks are not good.  It follows that
\[
\dist{\pbig}{\qbig} = \min_f \sum_{j=1}^{\numblk} d(B_j) \geq (\numblk/2000)(1/2500\numblk) = 1/(5 \cdot 10^6),
\]
as claimed.
\end{proof}

We are now in position to conclude the step~\ref{item:proof:overview:3} of our overall proof, establishing the following lower bound:
\begin{theorem}
  \label{theo:lower:bound:kvk}
	There exists an absolute constant $\dst_0>0$ such that, for every sufficiently large integers $\numparts, m \geq 1$, there exists $\q$ over $[\numparts$ for which testing $\binidentityproblem(\numparts,\q,\numparts,\dst_0)$ from ordered fingerprints only requires $\bigOmega{ (\numparts/m^2 2^m)^{1-1/m} }$ samples. (Moreover, one can take $\dst_0 = 10^{-7}$.)
\end{theorem}
\begin{proof}
Let $\pbig$ and $\qbig$ be the distributions described above, with $\numparts = \blksz \numblk$ and $\dst_0 \eqdef 10^{-7}$ as in~\cref{claim:dist:eps0}.  We consider the problem $\binidentityproblem(\numparts,\qbig,\numparts,\dst_0)$.  We will show that any algorithm that takes too few samples, and only considers their ordered fingerprint, can distinguish between $\pbig$ and $\qbig$ with probability at most $1/3$. Since $\totalvardist{\pbig}{\qbig} > \dst_0 \eqdef 10^{-7}$ (and, of course, $\totalvardist{\qbig}{\qbig} = 0$), this will establish our lower bound. 

By~\cref{lem:two-dist:same-moments}, an ordered fingerprint from $m$ samples from one of the $\numblk$ blocks of size $\blksz$ gives no information in distinguishing $\pbig$ and $\qbig$. Further, it is immediate to see, from the definition of ordered fingerprints, that the following holds as well:
\begin{fact}
\label{clm:equal-parts}
Let $\p_1$ and $\p_2$ be two distributions over $[\ab]$, such that the distribution of fingerprints for $\p_b$ is given by $\{\bF^b_j\}_{j \in [\ab]}$.  Suppose there exists a partition of $[\ab]$ into $(I_1,I_2,\ldots,I_m)$ such that the sequences $\{\bF^1_j\}_{j \in I_i}$ and $\{\bF^2_j\}_{j \in I_i}$ are identically distributed for all $i$ in $[m]$.  Then $\p_1$ and $\p_2$ cannot be distinguished by ordered fingerprints.
\end{fact}
This implies that if an algorithm sees at most $m$ samples from every block, then it cannot distinguish between $\pbig$ and $\qbig$ with ordered fingerprints from this sample.  It follows that any algorithm that distinguishes between $\pbig$ and $\qbig$ using ordered fingerprints must receive at least $m+1$ samples from at least one of the $\numblk$ many blocks.

To conclude, note that the distribution over the blocks is uniform; thus, by folklore collision estimates (see also Suzuki et al.~\cite{SuzukiTKT06}), the probability that one of the $\numblk$ blocks contains $m+1$ samples is at most $1/3$, given $O(\numblk^{1-1/(m+1)})$ samples.  Since $\numblk = \numparts/\blksz = \numparts/(5m^2 2^m)$, we get the claimed lower bound.    
\end{proof}
\noindent Setting $m = \sqrt{\log \numparts}$ in~\cref{theo:lower:bound:kvk}, we get:
\begin{corollary}
  \label{cor:lower:bound:kvk}
	There exists an absolute constant $\dst_0>0$ such that, for every sufficiently large integers $\numparts \geq 1$, there exists $\q$ over $[\numparts$ for which testing $\binidentityproblem(\numparts,\q,\numparts,\dst_0)$ from ordered fingerprints only requires $\frac{\numparts}{2^{O(\sqrt{\log\numparts})}}$ samples.
\end{corollary}
 
\bibliographystyle{alpha}
\bibliography{references}

\end{document}